\documentclass[11pt]{amsart}
\usepackage{a4wide,xcolor,eucal,enumerate,mathrsfs,graphics,subfig,booktabs,verbatim}
\usepackage{pdfsync}
\usepackage[hidelinks]{hyperref}
\usepackage{amsmath,amssymb,epsfig,amsthm,bm}
\usepackage[latin1]{inputenc}
\usepackage{dsfont}

\usepackage[english]{babel}
\usepackage{amsmath,amsfonts,amsthm,amssymb,dsfont,url}
\usepackage[foot]{amsaddr}
\usepackage[babel]{csquotes}
\usepackage{caption}
\usepackage[numbers,sort&compress]{natbib}
\bibpunct[, ]{[}{]}{,}{n}{,}{,}
\makeatletter
\def\NAT@def@citea{\def\@citea{\NAT@separator}}
\makeatother

\makeatletter
\g@addto@macro{\endabstract}{\@setabstract}
\newcommand{\authorfootnotes}{\renewcommand\thefootnote{\@fnsymbol\c@footnote}}%
\makeatother

\usepackage{epstopdf}

\theoremstyle{plain}
\newtheorem{theorem}{Theorem}[section]

\theoremstyle{definition}

\theoremstyle{remark}

\newtheorem{condition}[theorem]{Condition}

\begin{document}
\begin{center}
  \LARGE 
Isotonic regression for metallic microstructure data: estimation and testing under order restrictions \par \bigskip

  \normalsize
  \authorfootnotes
  Martina~Vittorietti\footnote{Contact author: m.vittorietti@tudelft.nl}\textsuperscript{1}\textsuperscript{2}, Javier~Hidalgo\textsuperscript{3}, Wei~Li\textsuperscript{2}\textsuperscript{3}, Jilt~Sietsma\textsuperscript{3}, Geurt~Jongbloed\textsuperscript{1}\par \bigskip

  \textsuperscript{1}Department of Applied Mathematics, Delft University of Technology\par
  \textsuperscript{2}Materials Innovation Institute (M2i)\par 
\textsuperscript{3}Department of Materials Science and Engineering, Delft University of Technology\par \bigskip

\end{center}


\date{}

\dedicatory{}



\begin{abstract}
Investigating the main determinants of the mechanical performance of metals is not a simple task. Already known physical inspired qualitative relations between 2D microstructure characteristics and 3D mechanical properties can act as the starting point of the investigation. Isotonic regression allows to take into account ordering relations and leads to more efficient and accurate results when the underlying assumptions actually hold. The main goal in this paper is to test order relations in a model inspired by a materials science application. The statistical estimation procedure is described considering three different scenarios according to the knowledge of the variances: known variance ratio, completely unknown variances, variances under order restrictions. New likelihood ratio tests are developed in the last two cases. Both parametric and non-parametric bootstrap approaches are developed for finding the distribution of the test statistics under the null hypothesis. 
Finally an application on the relation between Geometrically Necessary Dislocations and number of observed microstructure precipitations is shown.
\end{abstract}

\section{Introduction}
Understanding the intrinsic nature of the mechanical properties of metals is usually not an easy task. In order to get insight into what gives desired mechanical performance to a metal, a deep and detailed analysis of the metal microstructure characteristics is needed. For instance, it is known in literature that dislocations, i.e. line defects in the crystalline arrangement of the atoms  \cite{hull2001}, play a fundamental role in the mechanical behavior of metal alloys. More specifically, the appearance of Geometrically Necessary Dislocations\footnote{ Dislocations are usually classified into redundant and non-redundant dislocations, respectively called Statistically Stored Dislocations (SSDs) and Geometrically Necessary Dislocations (GNDs). GNDs are dislocations with a cumulative effect and they allow the accommodation of lattice curvature due to non-homogeneous deformation. They control the work hardening individually by acting as obstacles to slip and collectively by creating a long-range back stress.} (GNDs) during plastic deformation of the material contributes to the hardening of the material.
Detecting GNDs from 2D microstructure images is often challenging.
One widely accepted way is to use the so called Kernel Average Misorientation (KAM) \cite{moussa2017}. The KAM, measured in Electron BackScatter Diffraction (EBSD), quantifies the average misorientation around a measurement point with respect to a defined set of a nearest or nearest plus second-nearest neighbor points \cite{calcagnotto2010}.
In \cite{littlewood2011,revilla2014,bird2015} studies on the relation between GNDs and microstructure properties such as grain size and carbides size are presented. The relation between GNDs and grain size has both theoretical and experimental confirmation and it can be related to the well-known macroscopic physical Hall-Petch relation \cite{hall1951,petch1953}. In fact, the Hall-Petch relation, in its original version, describes the negative dependence of yield stress (mechanical property) on grain size; loosely speaking the smaller grains are, the stronger the material is. More specifically in \cite{kadkhodapour2011} the authors give as an explanation of the relation between GNDs and grain size that as the grain size decreases the grain boundary layer in which GNDs typically accumulate, occupies a greater volume fraction of the material, therefore it is reasonable to think that the smaller are the grains, the more GNDs will be observed.\\
Still unclear is instead the relation between carbides and GNDs. 
In fact, since the 1940's several studies on how carbides affect the mechanical behaviour of metals have been conducted. In \cite{pippel1999} the authors state that the primary carbides and their distribution have a major influence on the wear resistance and the toughness of the material. 
However, carbides tend to precipitate along the grain boundaries, that as said before, are the locations in which GNDs typically accumulate. Until now, no direct physical relationship has been found between carbides and Geometrically Necessary Dislocations.Therefore, isolating carbides effect and assessing the conjecture on the positive relation between carbides and GNDs is a problem of interest.\\
In \cite{hildago2019} a descriptive statistical analysis with response variable  KAM, used as a proxy of GNDs and as explanatory variables the number of grains, the number of carbides and the position of carbides revealed an almost monotone trend of the response variable according to the increments of the explanatories.\\
Therefore, in order to take into account the already known direction of the physical relation, we want to propose an approach that incorporates this information and a procedure for testing the prementioned conjectures on a new dataset.\\
In this context, isotonic regression comes to aid. In fact, the idea at the basis of isotonic regression is taking order restrictions into account for improving the efficiency of the statistical analysis by reducing the error or the expected error of estimates and increasing the power of the testing procedures, provided that the hypothesized order restriction actually holds.
The first papers about isotonic regression appeared in the 1950's \cite{ayer1955,vaneden1957} and books \cite{barlow1972,robertson1988} are well known references for statistical inference under order restrictions.
Isotonic regression proves its power in different fields such as epidemiology in testing the effects of different treatments or in dose-finding \cite{salanti2001,Stylianou2002}, but also in genetics \cite{Luss2012}, business \cite{Keshvari2013}, biology \cite{Barragan2013}.
There are not many examples of isotonic regression use in Materials Science. Throughout this paper special attention is given to the peculiar data structure.
Nowadays, developments towards multivariate isotonic regression, isotonic regression in inverse and censoring problems \cite{Guntuboyina2018,Groeneboom2014}, Bayesian isotonic regression \cite{li2017} are ongoing.
But also in the most basic framework there is still something missing.\\
In this paper, starting off with the most basic case, univariate isotonic regression of means under normality assumptions with known variances, we guide the reader into estimation and testing order restriction assumptions, considering different conditions on the variances.\\
Three different scenarios are considered. In all three cases, we focus on maximum likelihood as estimation procedure and likelihood ratio test as test statistic for hypothesis testing.\\
The first case is the basic case in which `the variances' are known or unknown but their ratio is known. This instance is considered extensively in \cite{barlow1972,robertson1988} and results for estimation and testing order restrictions are already known.\\
The second scenario is from an applications point of view the most common scenario in which the variances are unknown. In \cite{shi1998}, the authors derive a two steps estimating procedure for means and variances and interesting results on existence and uniqueness of the maximum likelihood estimates are derived under special conditions. Another iterative method, proposed in \cite{shi2008}, is extended to the unknown variances case.
The derivation of the test statistic and of its distribution in this scenario is not trivial. 
In fact, the estimate of the mean under the null hypothesis is also affected by the non knowledge of the variances. We propose the likelihood ratio test statistic and two different bootstrap approaches, one parametric and one non-parametric, for obtaining the test statistic distribution.\\
The last model considers not only the means under order restrictions but also the variances. This case has not often been faced probably because it is not common to have prior knowledge on the order of both means and variances. 
As in the unknown variances scenario, a two steps procedure for estimating means and variances is derived in \cite{shi1994} and similar results on existence and uniqueness under specific conditions on the empirical variances are given. 
In \cite{shi2008} an improved algorithm called Alternating Iterative Method (AIM) and more general results about convergence are derived.
For testing in this case we derive the likelihood ratio test taking into account the order of variances also under the null hypothesis and apply a parametric and non-parametric bootstrap approach in line with the one derived in the unknown variance case to obtain approximate p-values.\\
The paper structure is the following.
In Section \ref{sec:estimation} we explain the estimation procedure of the isotonic means in the three different cases (Sections \ref{subsec:isocase1}, \ref{subsec:isocase2}, \ref{subsec:isocase3}).
In Section \ref{sec:testing} the focus is on the Likelihood Ratio Test. We present it in the three different cases (Sections \ref{subsec:testcase1}, \ref{subsec:testcase2}, \ref{subsec:testcase3}) and in Section \ref{sec:bootstrap} we propose both a parametric and non-parametric bootstrap approach for approximating the distribution of the test statistics under the null hypothesis.
Finally, in Section \ref{sec:application} we come back to the application and we illustrate step-by-step how to deal with a real problem and more precisely how to perform isotonic regression and test for monotonicity of KAM with respect to the number of carbides.
The paper ends with conclusions in Section \ref{sec:conclusions}.

\section{Estimating restricted means in the normal case}
\label{sec:estimation}
We first introduce isotonic regression and the notation used in the rest of the paper in a more general context. Normality is assumed throughout this section.

Let $y_{ij}$, $j=1,\dots,n_i$, $i=1,\dots,k$ be the $j$th observation of the response variable $\mathrm{Y}$ corresponding to the $i$th level of the explanatory variable $\mathrm{X}$.\\
We assume $\mathrm{Y}_{ij}$ to be independent random variables, normally distributed with means $\mu_i$ and variances $\sigma_i^2$, $i=1,\dots,k$, $j=1,2,\dots,n_i$.\\
The log-likelihood is then given by
\begin{equation}
l(\bm{\mu},\bm{\sigma^2})=\sum_{i=1}^k\Big\{-\frac{n_i}{2}\mathrm{ln}\sigma_i^2-\frac{1}{2\sigma^2_i}\sum_{j=1}^{n_i}(y_{ij}-\mu_i)^2\Big\}+\mathrm{c}
\label{eq:likelihood}
\end{equation}
where $\mathrm{c}$ is a constant which does not depend on the parameters $\bm{\mu}=(\mu_1,\dots\mu_k)'$ and $\bm{\sigma^2}=(\sigma^2_1,\dots,\sigma^2_k)'$.\\
Furthermore, we assume that $\bm{\mu}$ satisfies
\begin{equation}
\mu_1\le\mu_2\le\dots\le\mu_k.
\label{eq:assumptionmean}
\end{equation}
A $k$-dimensional vector $\bm{\mathrm{\mu}}$ is said to be isotonic if $t\le s$ implies $\mu_t\le\mu_s$.\\
Be $D$ the set of all the isotonic vectors in $\mathbb{R}^k$,
\begin{equation}
D=\{\bm{\mu}\in \mathbb{R}^k; \mu_1\le\mu_2\le\dots\le\mu_k\}
\label{eq:d}
\end{equation}
In this section we are interested in the maximum likelihood estimator of $(\bm{\mu},\bm{\sigma^2})$, where $\bm{\mu}$ is isotonic and $\sigma^2_i>0$.
Depending on the information on $\bm{\sigma^2}$, different MLEs have been derived.\\
In the following three subsections the three different cases are considered.

\subsection{Isotonic regression of means with known variance ratio}
\label{subsec:isocase1}
This first case constitutes the most basic case in which all variances are either known or unknown but they differ according to some known multiplicative constants $c_i$.
This means that the variance $\sigma^2_i$ of the response variable $\mathrm{Y}_i$ is given by:
\begin{equation*}
\sigma^2_i=c_i\sigma^2,\,\,\,\, 1\le i\le k.
\end{equation*}
This specific case is already covered in \cite{barlow1972,robertson1988}, but we hereafter report the main results.
The problem of maximizing log-likelihood (\ref{eq:likelihood}) in $\mu$ can be rewritten equivalently as solving:
\begin{equation}
\min_{\mu\in D}\sum_{i=1}^k(\bar y_i-\mu_i)^2w_i
\label{eq:minloglik}
\end{equation}
 where $\bar y_i=\frac{\sum_jy_{ij}}{n_i}$ and $w_i=\frac{n_i}{c_i}$.
Note that this objective function does not depend on $\sigma^2$. 
The solution, $\bm{\hat \mu}^{I}$, is called the isotonic regression of $\bm{\bar y}=(\bar y_1,\dots,\bar y_k)$ with weights $\bm{w}=(w_1,\dots,w_k)$ \cite{shi1998}.
For obtaining the solution to (\ref{eq:minloglik}), different algorithms have been proposed in the literature (\cite{barlow1972},\cite{robertson1988}).
In this paper, the  ``Pool-Adjacent Violators Algorithm'' (PAVA) is used.\\
More details about the algorithm are provided in Appendix \ref{subsec:ALGORITHM 2.1}.\\

\subsection{Isotonic regression of means with unknown variances}
\label{subsec:isocase2}
In this second case, no assumptions on the variances are made.
They are unknown and for obtaining the maximum likelihood estimate of $\bm{\mu}$, they need to be estimated as well.
In \cite{shi1998} the authors consider this case and interesting results on existence and uniqueness of the MLE are achieved. We hereby recall the main results.
The approach is to maximize the log-likelihood (\ref{eq:likelihood}), with $\bm{\mu}\in D$ and $\bm{\sigma^2}\in\mathbb{R}^k_+$.\\
For any fixed $\bm{\sigma^2}\in \mathbb{R}^k_+$ the maximizer $\bm{\hat\mu^I}$ of $l(\bm{\mu},\bm{\sigma^2})$ over $\bm{\mu}\in D$ is the isotonic regression of $\bm{\bar y}$ with weights $\bm{w}=(w_1,\dots w_k)'$ and $w_i=\frac{n_i}{\sigma^2_i}$.\\
On the other hand, for any fixed $\bm{\mu}\in D$, the maximizer $\bm{\sigma^2}$ of $l(\bm{\mu},\bm{\sigma^2})$ over $\bm{\sigma}\in \mathbb{R}^k_+$
is $\bm{\hat\sigma^2}(\bm{\mu})=(\hat\sigma^2_1(\mu_1),\dots\hat\sigma^2_k(\mu_k))'$, where $\hat\sigma^2_i(\mu_i)=\frac{\sum_j^{n_i}(y_{ij}-\mu_i)^2}{n_i}$.\\
Substituting $\hat\sigma^2(\mu)$ into (\ref{eq:likelihood}), we can express the profile log-likelihood of $\bm{\mu}$ as

\begin{equation}
l(\bm{\mu})=\sum_{i=1}^k-n_i\mathrm{ln}[\bar\sigma^2_i+(\bar y_i-\mu_i)^2]+c
\label{eq:proflikelihoodmusigmaun}
\end{equation}
where $\bar\sigma^2_i=\frac{\sum_j^{n_i}(y_{ij}-\bar y_i)^2}{n_i}$ is the sample variance of the $i$th normal population and $c$ a constant that does not depend on $\bm{\mu}$.
Note that $l(\bm{\mu})\to-\infty$ if $\mu_k\to\infty$ or $\mu_1\to-\infty$.
Hence, maximizing $l$ over $D$ is equivalent to maximizing $l$ over a compact subset of $D$ of type $D_a=\{\mu\in D:\mu_1\ge-a,\mu_k\le a\}$. As $l$ is continuous on $D_a$, a maximizer over $D$ exists.\\
As previously said, the authors in \cite{shi1998} discuss also uniqueness of the MLE of $(\bm{\mu},\bm{\sigma^2})$.
They state that $l$ is not a concave function in general and that for guaranteeing uniqueness the following condition suffices (see Theorem 2.3 \cite{shi1998}):
\begin{condition}
For $i=1,\dots,k$, $\bar\sigma^2_i>\max\{(\bar y_i-\min(\bm{\bar y}))^2,(\bar y_i-\max (\bm{\bar y}))^2\}$.
\label{conditiona}
\end{condition}
For finding a maximizer of (\ref{eq:proflikelihoodmusigmaun}), a two steps iterative algorithm based on PAVA has been proposed in \cite{shi1998}.
From an initial guess for $\bm{\mu}$, the associated maximizer in $\bm{\sigma^2}$ is computed and after that the maximizer in $\bm{\mu}$ based on this $\bm{\sigma^2}$ and so on. This iterative procedure stops when the maximum difference between the estimated means at step $l-1$ and at step $l$ is less than an arbitrary small threshold value, e.g.,
\begin{equation*}
\max_{1\le i\le k}|\mu_i^{I(l-1)}-\mu_i^{I(l)}|\le 10^{-m},
\end{equation*}
where $m$ is taken to be equal to $3$ in our case.
In \cite{shi2008} the authors propose a new algorithm called Alternating Iterative Method (AIM).
The procedure is based on the minimization of a semi-convex function.
In particular, restating the problem in terms of $(\bm{\mu},\bm{\nu})$, where $\bm{\nu}=(1/\sigma^2_1,\dots,1/\sigma^2_k)'$ and given $D_a$ is a convex subset of $\mathbb{R}^k$ and $V$ a convex subset of $\mathbb{R}^k_+$, $V=\{\nu\in\mathbb{R}^k_+:0\le 1/\max_i(\min_{\min(\bm{\bar y})\le\theta\le\max(\bm{\bar y})}s^2_i(\theta))\le \nu_i\le1/ \min_i(\min_{\min(\bm{\bar y})\le\theta\le\max(\bm{\bar y})}s^2_i(\theta))\}$, $L(\bm{\mu},\bm{\nu})$ is a semi-convex function because: i) $L(\bm{\mu},\bm{\nu})$ is defined on $D_a\times V$; ii) for any given $\bm{\mu}\in D_a$, $L(\bm{\mu},\cdot)$ is strictly convex on $V$ and, for any given $\bm{\nu}\in V$,  $L(\cdot,\bm{\nu})$ is strictly convex on $D_a$.
The algorithm originally proposed for the simultaneous order restrictions of means and variances can be easily extended to the unknown variance case.
The iteration method works in alternating the search of the minimum point, $\bm{\mu}^{(l)}$, of $L(\bm{\mu},\bm{\nu}(\bm{\mu}^{(l-1)}))$ on a compact subset $D_a$ and the search of the minimum point, $\bm{\nu^{(l)}}$, of $L(\bm{\mu}(\bm{\nu}^{(l-1)}),\bm{\nu})$ on $V$.
Proof of the convergence of the algorithm does not require additional conditions \cite{shi2008}.
The iterative procedure stops when the difference between the likelihoods at step $l-1$ and at step $l$ is less than an arbitrary small threshold value:
\begin{equation}
|L(\bm{\mu^{(l-1)}},\bm{\nu^{(l-1)}})-L(;\bm{\mu^{(l)}},\bm{\nu^{(l)}})|\le 10^{-m}
\label{eq:stopcriteria}
\end{equation}
A more detailed version of both algorithms is reported in Appendix \ref{subsec:ALGORITHM 2.2}.

\subsection{Isotonic regression of means and variances simultaneously}
\label{subsec:isocase3}
We now assume that both mean and variances are restricted by simple orderings. Therefore, in addition to assumption (\ref{eq:assumptionmean}), we assume also:
\begin{equation}
\sigma_1^2\ge\sigma_2^2\ge\dots\ge\sigma_k^2>0
\label{eq:assumptionvar1}
\end{equation}
The reason for taking decreasing order is relates to our application considered in Section \ref{sec:application}; increasing variances can be dealt with analogously.
In \cite{shi1994}, maximum likelihood estimation under simultaneous order restrictions on mean and variances from a Normal population is studied.
Some of the most important results are hereby recalled.
The approach is to maximize the log-likelihood (\ref{eq:likelihood}) with $\bm{\mu}\in D$ and $\bm{\sigma^2}\in \bar G$, where $\bar G$ is the closure of
\begin{equation}
G=\{\bm{\sigma^2}\in\mathbb{R}^{k}_+:\sigma_1^2\ge\sigma_2^2\ge\dots\ge\sigma_k^2>0\}.
\end{equation}
This means that the maximizer will have positive $\sigma^2$-values if there is variation within the groups.
Then, for any fixed $\bm{\sigma^2}\in G$, the maximizer $\bm{\mu^I}$ of $l(\bm{\mu},\bm{\sigma^2})$ over $\bm{\mu}\in D$ is the isotonic regression of $\bm{\bar y}$ with weights $\bm{w}=(w_1,\dots w_k)'$ and $w_i=\frac{n_i}{\sigma^2_i}$.\\
Furthermore, for any $\bm{\mu} \in D$, the maximizer $\bm{\hat\sigma^{2I}}(\bm{\mu})$ of $l(\bm{\mu}, \bm{\sigma^2})$ is the so called antitonic regression (isotonic regression with reversed order \cite{Groeneboom2014}) of $\bm{s^2}=(s_1^2,\dots s_k^2)'$, $s_i^2=\frac{\sum_{j=1}^{n_i}(y_{ij}-\mu_i)^2}{n_i}$, with weights $\bm{N}=(n_1,\dots,n_k)'$.
Existence is guaranteed noticing that \\$\bm{\sigma^2}\in [\min_i(\min_{\min_i(\bm{\bar y})\le\theta\le\max_i(\bm{\bar y})}s^2_i(\theta)),\max_i(\min_{\min_i(\bm{\bar y})\le\theta\le\max_i(\bm{\bar y})}s^2_i(\theta))]$, $s^2_i(\theta)=\sum_{j=1}^{n_i}(y_{ij}-\theta)^2/n_i$ (see Theorem 2.1 \cite{shi1994}).\\
Uniqueness is proved under the following condition (see Theorem 2.2 \cite{shi1994})
\begin{condition}
For $i=1,\dots k$ the sample variance $\bar\sigma_i^2$ satisfies $\bar\sigma^2_i>2(b-a)$ where $b$ and $a$ are the maximal and the minimal means respectively.
\label{conditionb}
\end{condition}
As in the unknown variances case, a two steps iterative algorithm is proposed for finding the solution for both means and variances under order restrictions.
The proof of the convergence of the algorithm is given under Condition \ref{conditionb}. \\
Later, in \cite{shi2008}, as mentioned in the previous section, the authors show that restating the problem in terms of $(\bm{\mu},\bm{\nu})$, where $\bm{\nu}=(1/\sigma^2_1,\dots,1/\sigma^2_k)'$ Condition \ref{conditionb} is not needed for proving that the algorithm converges. In fact, also in this case the proposed AIM algorithm can be employed.
Since $L(\bm{\mu},\bm{\nu})$ has continuous second-order partial derivatives and the Hessian matrix with respect to $\mu$ $H(\bm{\mu},\bm{\nu})=\mathrm{diag}(n_1\nu_{1},\dots,n_k\nu_{k})$ is a positive definite diagonal matrix for any fixed $\bm{\nu}=(\nu_{1},\dots,\nu_{k})'\in V_0$, $V_0=\{\nu\in\mathbb{R}^k:0\le 1/\max_i(\min_{\min_i(\bm{\bar y})\le\theta\le\max_i(\bm{\bar y})}s^2_i(\theta))\le \nu_1\le\dots\le\nu_k\le1/ \min_i(\min_{\min_i(\bm{\bar y})\le\theta\le\max_i(\bm{\bar y})}s^2_i(\theta))\}$ then  by Theorem 4 in \cite{shi2008} the iterative sequence of solutions to $L(\bm{\mu},\bm{\nu})$, $\{(\bm{\mu^{(n)}},\bm{\nu^{(n)}})\}$ converges to the MLE  solution and consequently the sequence $\{(\bm{\mu^{(n)}},\bm{\sigma^{2(n)}})\}$ as well.\\ 
As in the previous case, the alternating iterative procedure is stopped when the maximum difference between the likelihoods at step $l-1$ and at step $l$ is less than an arbitrary small threshold value (see (\ref{eq:stopcriteria})).\\
A pseudo code of the algorithms can be found in Appendix \ref{subsec:ALGORITHM 2.3}.

\section{Likelihood Ratio Test: constant $\mu$ against monotonicity}
\label{sec:testing}
We are interested in testing hypotheses of monotonicity in $\mu$ under the various assumptions on the variances discussed in Section \ref{sec:estimation}. 
There exists extensive literature on testing hypotheses on means.
In most cases, a standard testing procedure entails testing the hypothesis of equality of means against the hypothesis that they are different. In this paper, we consider the same null hypothesis but the alternative is different: monotonicity of the means.
As in the previous section, we consider three different testing frameworks according to the different assumptions on the variances.
In all three different scenarios the test statistic of interest is the Likelihood Ratio Test (LRT), an intuitive and powerful tool in hypothesis testing.
In both \cite{barlow1972} and \cite{robertson1988} an entire chapter is dedicated to LRT developments and its use for testing order restrictions hypothesis under the normality assumption and known variance ratio.
Using the same notation used in Section \ref{sec:estimation}, we wish to test
\begin{equation*}
H_0:\,\,\,\,\,\mu_1=\mu_2=\dots=\mu_k
\end{equation*}
against monotonicity of means 
\begin{equation}
H_1:\,\,\,\,\,\mu_1\le\mu_2\le\dots\le\mu_k.
\label{H1}
\end{equation}
The likelihood ratio test for $H_0$ against $H_1$ can be defined as:
\begin{equation}
\Lambda=\frac{\max_{(\bm{\mu}\in H_0;\bm{\sigma^2})}L(\bm{y_1},\bm{y_2},\dots,\bm{y_k};\bm{\mu},\bm{\sigma^2})}{\max_{(\bm{\mu}\in H_1;\bm{\sigma^2})}L(\bm{y_1},\bm{y_2},\dots,\bm{y_k};\bm{\mu},\bm{\sigma^2})}
\label{eq:genericlrt}
\end{equation}
where $\bm{y_i}=(y_{i1},\dots,y_{in_i})'$, $\bm{\mu}=(\mu_1,\dots\mu_k)'$ 
and $\bm{\sigma^2}=(\sigma^2_1,\dots\sigma^2_k)'$ . It rejects the null hypothesis for small values of $\Lambda$ or alternatively for large values of $-2\log\Lambda$. The convenience in using this other form lies on the analogy with the $\chi^2$ statistic used to test against the alternative hypothesis $\bar H_0$, that not all $\mu_i$'s, $i=1,\dots,k$, are the same.\\
In the following subsections more explicit expressions for $\Lambda$ are given depending on the specific assumptions on means and variances.

\subsection{LRT with known variance ratio}
\label{subsec:testcase1}
As in Section \ref{subsec:isocase1} let $y_{ij}$ $j=1,2,\dots n_i$, $i=1,2,\dots k$  be independent observations, normally distributed with unknown mean $\mu_i$ and variances $\sigma^2_i=c_i\sigma^2$ with $c_i$ known and $\sigma^2$ unknown.
Under $H_0$, the maximum likelihood estimate of $\mu_1=\mu_2=\dots=\mu_k$ is given by:
\begin{equation}
\hat \mu_{H_0}=\frac{\sum_{i=1}^kw_i\bar y_i}{\sum_{i=1}^kw_i}
\end{equation}
with $w_i=\frac{n_i}{c_i}$.
Under $H_1$ the MLE of $\bm{ \mu}$ is $\bm {\hat \mu^I_{H_1}}$, the isotonic regression of $\bm{\bar y}$, with weights $\bm{w}=(w_1,\dots,w_k)'$, with respect to the simple order defined in (\ref{H1}).\\
The likelihood ratio test for $H_0$ against $H_1$, if the variances are known and $c_i=1$ boils down to rejecting $H_0$ for large values of

\begin{equation}
-2\log\Lambda=\frac{1}{\sigma^2}\Big[\sum_{i=1}^k\sum_{j=1}^{n_i}(y_{ij}-\hat\mu_{H_0})^2-\sum_{i=1}^k\sum_{j=1}^{n_i}(y_{ij}-\hat\mu_{iH_1}^I)^2\Big]
\label{eq:-2logl}
\end{equation}
It is easy to check that the test is equivalent to rejecting $H_0$ for large values of:
\begin{equation}
\bar \chi^2=\frac{\sum_{i=1}^k\bar\chi^2_i}{\sigma^2}
\end{equation}
where $\bar\chi^2_i=n_i(\hat\mu_{iH_1}^I-\hat\mu_{H_0})^2$ and $\sigma^2$ is the (known) common value of the variance.

Now, let us consider the more general case, $\sigma^2_i=c_i\sigma^2$ with $c_1,c_2,\dots c_k$ known and $\sigma^2$ unknown.
The estimator of $\sigma^2$ under the null hypothesis is
\begin{equation}
\hat\sigma_{H_0}^2=\frac{\sum_{i=1}^kc_i^{-1}\sum_{j=1}^{n_i}(y_{ij}-\hat\mu_{H_0})^2}{N}
\end{equation}
and under $H_1$
\begin{equation}
\hat\sigma_{H_1}^2=\frac{\sum_{i=1}^kc_i^{-1}\sum_{j=1}^{n_i}(y_{ij}-\hat\mu_{iH_1}^I)^2}{N}
\end{equation}
The likelihood ratio test rejects $H_0$ for small values of $\Lambda=\big(\frac{\hat\sigma_{H_1}^2}{\hat\sigma^2_{H_0}}\big)^{N/2}$ or equivalently, taking $\bar E^2=1-\Lambda^{2/N}$, for large values of 
\begin{equation}
\bar E^2=\frac{\sum_{i=1}^kc_i^{-1}\bar\chi_i^2}{\sum_{i=1}^kc_i^{-1}\sum_{j=1}^{n_i}(y_{ij}-\hat\mu_{H_0})^2}
\end{equation}
An extension to the multivariate case with covariance matrix $\Sigma$ unknown but common can be found in \cite{perlman1969,sasabuchi2007}.

\subsection{LRT with unknown variances}
\label{subsec:testcase2}

In this second case, no assumptions on the variances are made. They are unknown and possibly unequal. 
Using the notation of Section \ref{subsec:isocase2} let $y_{ij}$, $j=1,2,\dots,n_i$, $i=1,2,\dots,k$ be independent observations from a univariate Normal distribution with unknown mean vector $\mu_i$ and completely unknown variances $\sigma^2_i>0$.
Let $\bm{\hat\mu^I}$ be the solution of the isotonic regression of $\bm{\bar y}$ with weights $\bm{w}=(w_1,\dots,w_k)'$, $w_i=\frac{n_i}{\sigma^2_i}$ found used Algorithm (2.2) in Appendix \ref{sec:appendix}.\\
The first example of testing when all the variances are unknown can be found in \cite{bazyari2017} and the univariate version of the test proposed by the author is:
\begin{equation}
\sum_{i=1}^k\frac{(\hat\mu_i^I-\bar y)^2n_i}{s^2_i}
\end{equation}
where $\bar y=\frac{\sum_{i=1}^k n_i\bar y_i}{\sum_{i=1}^kn_i}$ and $s^2_i=\frac{\sum_{i=1}^k\sum_{j=1}^{n_i}(y_{ij}-\bar y_i)^2}{n_i-1}$.
This test is clearly inspired by the LRT but it is not. \\
Let us consider first the maximum likelihood solution $(\hat\mu_{H_0},\bm{\hat\sigma^2}_{H_0})$, $\bm{\hat\sigma^2_{H_0}}=(\hat\sigma^2_{1H_0},\dots,\hat\sigma^2_{kH_0})'$ under the null hypothesis.
The log-likelihood under the null hypothesis is
\begin{equation}
l(\mu,\bm{\sigma^2})=\sum_{i=1}^k\Big\{-\frac{n_i}{2}\mathrm{ln}\sigma_{i}^2-\frac{1}{2\sigma^2_{i}}\sum_{j=1}^{n_i}(y_{ij}-\mu)^2\Big\}+\mathrm{c}.
\label{eq:likelihoodh0uv}
\end{equation}
Differentiating this loglikelihood with respect to $\mu$ and $\sigma^2_{i}$, the  following $k+1$ score equations in $k+1$ unknowns emerge:
\begin{equation}
\begin{cases}
 \mu_{H_0}=\frac{\sum_{i=1}^k n_i\sigma^{-2}_{iH_0} \bar y_i}{\sum_{i=1}^k n_i\sigma^{-2}_{iH_0}} & \\[15pt]

 \sigma^2_{iH_0}=\sum_{j=1}^{n_i}n_i^{-1}(y_{ij}-\mu_{H_0})^2 & 1\le i\le k
\end{cases}
\label{eq:mlesolcase2}
\end{equation}
Substituting $\bm{\sigma^2_{H_0}(\mu)}$ in (\ref{eq:likelihoodh0uv}), the profile likelihood of $\mu$ is:
\begin{equation}
l(\mu)=-\sum_{i=1}^k\frac{n_i}{2}\mathrm{ln}\big(\sum_{j=1}^{n_i}n_i^{-1}(y_{ij}-\mu)^2\big)+c.
\label{eq:profmuuv}
\end{equation}

\begin{theorem}
A maximizer of (\ref{eq:profmuuv}) over $\mathbb{R}^d$ exists and it is contained in $[\min_i\bar y_i,\max_i\bar y_i]$. Moreover, if  $[\min_i\bar y_i,\max_i\bar y_i]\in[\max_{1\le i\le k }(\bar y_i -\bar \sigma_i),\min_{1\le i\le k}(\bar y_i +\bar \sigma_i)]$ then the maximizer is unique.
\label{thm:exisuniquv}
\end{theorem}

\begin{proof}
Maximizing profile likelihood of $\mu$ (\ref{eq:profmuuv}) boils down to maximize the sum of functions

\begin{equation}
-\frac{n_i}{2}\mathrm{ln}(n_i(\bar\sigma^2_i+(\bar y_i-\mu)^2)), \,\,\,\, i=1,\dots,k.
\label{eq:prototype}
\end{equation}
Functions of type (\ref{eq:prototype}) are unimodal with mode at $\bar y_i$ and strictly concave on $[\bar y_i-\bar\sigma_i; \bar y_i+\bar\sigma_i]$. As the sum of unimodal functions is decreasing to the right of the rightmost mode (since all terms are decreasing) and from $-\infty$ to the leftmost mode, the sum is increasing (as all of the functions are increasing on that set).
Therefore, any maximizer of $l$, if it exists, belongs to the interval $[\min_i \bar y_i,\max_i\bar y_i]$. As $l$ is continuous on $[\min_i \bar y_i,\max_i\bar y_i]$, existence of a maximizer is guaranteed.\\
Then if we consider the (possibly empty) interval where all the functions in (\ref{eq:prototype}) are strictly concave, on that interval the sum is also strictly concave.
As for each $i$ the function (\ref{eq:prototype}) is strictly concave on $I_i=[\bar y_i-\bar\sigma_i;\bar y_i-\bar\sigma_i]$, (\ref{eq:profmuuv}) is strictly concave on $\bigcap_{i=1}^k I_i$. If $[\min_i\bar y_i,\max_i\bar y_i]$ is contained in this intersection, $l$ is strictly concave on $[\min_i\bar y_i,\max_i\bar y_i]$. Hence $l$ has a unique maximizer on $\mathbb{R}^d$.\\
\end{proof}
\textit{Remark}: in a setting with real data, it is easy to check whether $[\min_i\bar y_i,\max_i\bar y_i]\in[\max_{1\le i\le k }(\bar y_i -\bar \sigma_i),\min_{1\le i\le k}(\bar y_i +\bar \sigma_i)]$ and hence to determine whether the maximum is unique.\\
However, as seen from (\ref{eq:mlesolcase2}) the MLE estimate $(\mu,\bm{\sigma^2})$ has no closed form expression. 
Therefore, in \cite{gokpinar2012} and \cite{mutlu2017} two different methods for finding the optimal solution are proposed.
The first is an iterative procedure based on the Newton-Raphson method.
A reasonable initial value for $\hat \mu^{(0)}_{H_0}$ is the so called Graybill-Deal estimator \cite{graybill1959} $\hat \mu_{(GD)}=\frac{\sum_{i=1}^k(n_i\bar y_i)/\bar s^2_i}{\sum_{i=1}^kn_i/ \bar s^2_i}$ with $\bar s^2_i=\frac{\sum_{j=1}^{n_i}(y_{ij}-\bar y_i)^2}{n_i-1}$. The convergence speed of the algorithm strongly depends on the initial values. 
The second method is based on the profile likelihood approach.
The authors in \cite{mutlu2017} propose the bisection method for finding the zero of the profile likelihood with respect to $\mu_{H_0}$.
Under $H_1$ we use as estimates of $(\bm{\mu_{H_1}},\bm{\sigma^2_{iH1}})$, $(\bm{\hat\mu^I},\bm{\hat\sigma^2})$ found using the iterative procedure described in Section \ref{subsec:isocase2}.\\
The likelihood ratio test when the variances are completely unknown can be expressed as:
\begin{equation*}
\tilde \Lambda=\prod_{i=1}^k\Big(\frac{\hat\sigma^2_{iH_0}}{\hat\sigma^2_{iH_1}}\Big)^{-\frac{n_i}{2}}
\end{equation*}
Therefore, as in the previous case, the test rejects for small values of $\tilde \Lambda$ or equivalently for large values of $-2\log\tilde\Lambda$.

\subsection{LRT with ordered variances}
\label{subsec:testcase3}
Using the notation of Section \ref{subsec:isocase3} let $y_{ij}$, $j=1,2,\dots,n_i$, $i=1,2,\dots,k$ be independent observations from Normal distributions with mean vector $\mu_i$ and variances $\sigma^2_i$. 
As in the previous case, the first step is the estimation of $(\mu,\bm{\sigma^2})$ under the null hypothesis.
In this case we need to maximize (\ref{eq:likelihoodh0uv})
under the restriction

\begin{equation}
\sigma_1^2\ge\sigma_2^2\ge\dots\ge\sigma_k^2>0.
\label{eq:sigmaconstraints}
\end{equation}

\begin{theorem}
Suppose that for $1\le i\le k$, $\bar{\sigma}_i^2>0$. Then there exists a maximizer of (\ref{eq:likelihoodh0uv}) under constraints (\ref{eq:sigmaconstraints}).
\label{thm:exisuniqov}
\end{theorem}
\begin{proof}
First consider the situation for fixed $\bm{\sigma^2}$ with $\sigma_i^2>0$ for all $i$. Differentiating (\ref{eq:likelihoodh0uv}) with respect to $\mu$ yields the equation
$$
\sum_{i=1}^k \frac{n_i(\bar{y_i}-\mu)}{\sigma^2_i}
$$
This shows, that for this $\bm{\sigma^2}$, the (unique) maximizer  of (\ref{eq:likelihoodh0uv}) in $\mu$ is given by the following weighted sum of level-means,
$$
\hat{\mu}(\bm{\sigma^2})=\frac{\sum_{i=1}^kn_i\sigma^{-2}_{i} \bar y_i}{\sum_{i=1}^kn_i\sigma^{-2}_{i}}
$$
Consequently, $\min_i\bar{y}_i\le\hat{\mu}(\bm{\sigma^2})\le\max_i\bar{y}_i$, bounding the set of possible maximizers of (\ref{eq:likelihoodh0uv}) in $\mu$ irrespective of the precise value of $\bm{\sigma^2}$.\\
Now, given any $\mu\in\mathbb{R}$, the corresponding optimal $\bm{\sigma^2}$ is the solution to the  antitonic regression problem $\mathrm{antireg}(\bm{\bar\sigma^2_{H_0}},\bm{N})$ where $\bm{\bar\sigma^2}=(\bar\sigma_{1}^2,\dots \bar\sigma_{k}^2)'$, $\bar\sigma_{i}^2=\frac{\sum_{j=1}^{n_i}(y_{ij}-\mu)^2}{n_i}$, $\bm{N}=(n_1,\dots,n_k)'$ (see \cite{robertson1988} Example 1.5.5).
The vector to be projected has elements $\bar\sigma_i^2+(\mu-\bar y_i)^2$. This means, that if $\mu$ is restricted to $[\min_i\bar{y}_i,\max_i\bar{y}_i]$, the coordinates to be projected all belong to the interval
 $[\min_i\bar\sigma_i^2,\max \bar\sigma_i^2+(\max \bar y_i-\min \bar y_i)^2]$. So, if $\mu$ ranges over $[\min_i\bar{y}_i,\max_i\bar{y}_i]$, the optimal $\bm{\sigma^2}$ is also contained in a the closed bounded region $[\min_i\bar\sigma_i^2,\max \bar\sigma_i^2+(\max \bar y_i-\min \bar y_i)^2]^k$. By our assumption that all $\bar\sigma_i^2>0$, the MLE exists being a maximizer of a continuous function on a compact set in $\mathbb{R}\times \mathbb{R}^k$
\end{proof}
If we consider this case as a special case of the case considered in \cite{shi1994} the solution is unique if Condition \ref{conditionb} holds. 
Given that the solution is not in a closed form, we use an iterative procedure to approximate the solution.
As a starting value $\hat\mu^{(0)}$, a modified version of the Graybill-Deal estimator of the common mean when the variances are subject to order restrictions proposed in \cite{misra1997} appears to be a good choice:
\begin{equation}
\hat\mu_{(I)}=\frac{\sum_{i=1}^kw_i\hat\tau_i\bar y_i}{\sum_{i=1}^kw_i\hat\tau_i}
\end{equation}
where $\hat\tau_i$ is the isotonic regression of $(\bm{t},\bm{N})$ where $\bm{t}=(t_1,\dots t_k)'$, $t_i=\frac{1}{s_i^2}$.\\
Under $H_1$ we use as estimates of $(\bm{\mu_{H_1}},\bm{\sigma^2_{iH1}})$, $(\bm{\hat\mu^I},\bm{\hat\sigma^{2I}})$ found using the iterative procedure described in Section \ref{subsec:isocase3}.\\
In contrast with the previous cases, it is not possible to further reduce the expression of the LRT because
\begin{equation*}
\exp\Big\{\frac{1}{2}\sum_{i=1}^k\sum_{j=1}^{n_i}\frac{(y_{ij}-\hat\mu_{H_0})^2}{\sigma^2_{H_0}}\Big\}
\end{equation*}
does not reduce to a constant. The same holds under $H_1$.
Therefore the LRT in this case can be computed by substituting the solutions obtained via the iterative procedure under $H_0$ and $H_1$ in the generic expression given in (\ref{eq:genericlrt}):

\begin{equation}
\Lambda^I=\frac{L(\hat\mu_{H_0},\bm{\hat\sigma^{2I}_{H_0}})}{L(\bm{\hat\mu^I_{H_1}},\bm{\hat\sigma^{2I}_{H_1}})}
\label{eq:lrtordvar}
\end{equation}

\section{Null hypothesis distribution of the test statistics: bootstrap approach}
\label{sec:bootstrap}
In order to determine the significance of the various test statistics proposed in the previous sections, we need the null hypothesis distribution of the test statistics.
The main distributional results concerning $\bar\chi^2_k$ and $\bar E^2_k$, the test statistics derived in the known variance ratio case, are contained in \cite{barlow1972} (Theorems 3.1-3.2).
However, problems related to the value of $k$ can arise in the analytical derivation of the p-values.
Numerical approximation can be necessary, especially if $k>4$ and if the variation in the range of the weights is not `moderate' \cite{robertson1983,silvapulle2005}.\\
Furthermore, in the case of completely unknown variances, the null distribution depends on the unknown variances. 
When analytical derivation of the null distribution is particularly complex or not possible, bootstrap methodology is a good option. 
Therefore, we propose both a parametric and a non-parametric bootstrap approach that can be easily employed for finding approximate p-values taking into account the different assumptions on the variances.
For overcoming the complex derivation when the variances are unknown, bootstrap procedures have been proposed in the literature \cite{minhajuddin2007,bazyari2017}.\\  
In particular, in \cite{minhajuddin2007} an interesting review of the methods used to approximate the null distribution of the test statistic under $H_0$ and the restrictive normality assumption (with which we will not deal in this paper) is reported. Moreover the authors propose both a parametric and non-parametric bootstrap approach for the likelihood ratio test null distribution for one sided hypothesis testing for means in a multivariate setting \cite{minhajuddin2007}.
Also in \cite{bazyari2017} a bootstrap approach to test the homogeneity of order restricted mean vectors when the covariance matrices are unknown is used.
In line with those previous approaches, here we propose two general bootstrap procedures, parametric and non-parametric, that can be used for testing the null hypothesis taking into account the various assumptions on the variances.

\subsection*{Parametric bootstrap}
\textbf{Algorithm}:\\
 \begin{description}
 \item[(1)] Obtain the estimates  $\bm{\hat\mu^I_{iH_1}}$ and $\hat\mu_{H_0}$ using the original data and compute the observed value of the test statistic of interest $LRT^{(0)}$ ($\bar\chi^{2(0)}$, $\bar E^{2(0)}$, $\tilde \Lambda^{(0)}$ or $ \Lambda^{I(0)}$).
 \item[(2)] Generate, for $1\le i\le k$, $1\le j\le n_i$ $\bm{Y_{ij}^*}\sim N\big(\hat\mu_{H_0},\sqrt{\sigma^2_{iH_0}}\big)$, independently.
 \item[(3)] For $(\bm{Y_{i}^*},\dots,\bm{Y_{k}^*})$ obtain the estimates $\bm{\hat\mu^{I*}_{i}}$ and $\hat\mu^*$ and compute the bootstrap test statistic of interest $LRT^*$
 \item[(4)] Repeat \textbf{(2)}-\textbf{(3)} for a sufficient large number of times $M$
 \end{description}
The bootstrap approximation of the p-value is the given by:
\begin{equation}
p\approx\frac{\#(LRT^*>LRT^{(0)})}{M}
\label{eq:bootpvalue}
\end{equation}
 and the null hypothesis is rejected whenever this p-value is less than the nominal level $\alpha$.\\
Step (2) is the key step, in which the assumption on the variances play a crucial role.
It is interesting to notice that the above procedure can be further simplified.
In fact, we can instead of generating individual observations, directly generate empirical means $\bar y_i=Z_i*\frac{\sigma}{\sqrt{n_i}}$, with $Z_i$ Standard Normally distributed.

\subsection*{Non-parametric bootstrap}
 The non-parametric version of the bootstrap releases the normality assumption of the bootstrap samples. 
However a relatively large sample size is required for the following approach. 

 \textbf{Algorithm}:\\
 \begin{description}
  \item[(1)] Obtain the estimates $\bm{\hat\mu^I_{iH_1}}$ and $\hat\mu_{H_0}$ using the original data and compute the observed value of the test statistic of interest $LRT^{(0)}$ ($\bar\chi^{2(0)}$, $\bar E^{2(0)}$, $\tilde \Lambda^{(0)}$ or $ \Lambda^{I(0)}$). 
 \item[(2)] Standardize the original $y_{ij}$ to obtain `standardized residuals' $z_{ij}=\frac{y_{ij}-\bar y_i}{\sqrt{s^2_i}}$
 \item[(3)] Combine all $z_{ij}$ observations from $(1\le i\le k; 1\le j\le n_i)$  into a vector of length $\sum_{i=1}^k n_i$ and draw $k$ simple random samples $z^*_{ij}$ with replacement each of respective sizes $(n_1,n_2,\dots,n_k)$
 \item[(4)] Transform $z^*_{ij}$ to $y_{ij}^*=z_{ij}^*\cdot\tilde{\sigma}_{iH_0}+\hat\mu_{H_0}$
 \item[(5)] For each bootstrap sample $y^*_{ij}$ $(1\le i\le k; 1\le j\le n_i)$ obtain the estimates  $\bm{\hat\mu^{I*}_{iH_1}}$ and $\hat\mu_{H_0}^*$ and compute the bootstrap test statistic of interest  $LRT^*$
 \item[(6)] Repeat  \textbf{(3)}-\textbf{(4)}-\textbf{(5)} for a sufficient large number of times $M$
 \end{description}
 
The bootstrap approximated p-value is defined as in the parametric case (\ref{eq:bootpvalue}).

\section{Application}
\label{sec:application}
One of the most common ways for investigating strength and ductility of metallic materials is by performing a tensile test.
Loosely speaking a tensile test is an experiment in which force is applied to the test sample causing deformation of the material, temporarily (elastic behavior), permanently (plastic behavior) and eventually its fracture \cite{Davis2004}.
Data used in this paper are image data of
the microstructure of the material subjected to a plastic strain (deformation) of 0.139 obtained performing a uniaxial tensile test in which force is applied to the test sample with respect to just one specific axis (Fig. \ref{fig:tensiletest}).
\begin{figure}[!h]
\centering
\includegraphics[scale=0.06]{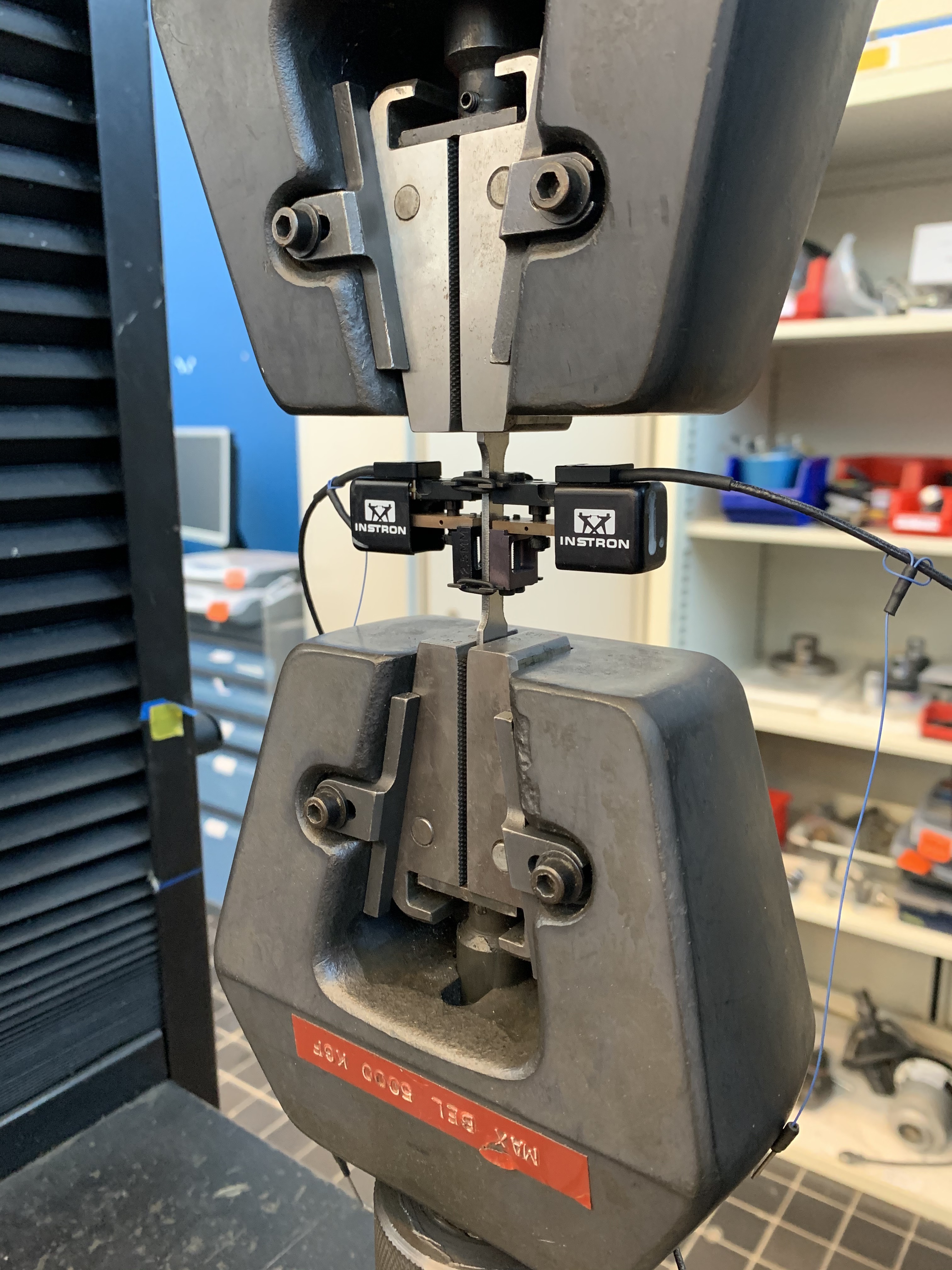}
\caption{Tensile testing machine}
\label{fig:tensiletest}
\end{figure}
At a microstructure level the deformation of the material corresponds to displacements in the lattice structure and in the possible appearance of Geometrically Necessary Dislocations (GNDs).
The material used in this paper is an annealed AISI420 stainless steel with $M_{23}C_6$ carbides and aim is investigating the carbide effect on the GNDs formation. 
Kernel Average Misorientation (KAM) is used as a proxy of the GNDs. In Figure \ref{fig:kam25} the KAM is represented by red filaments, the blue lines represent the grain boundaries, carbides are the black dots.
\begin{figure}[!h]
\centering
\includegraphics[scale=0.15]{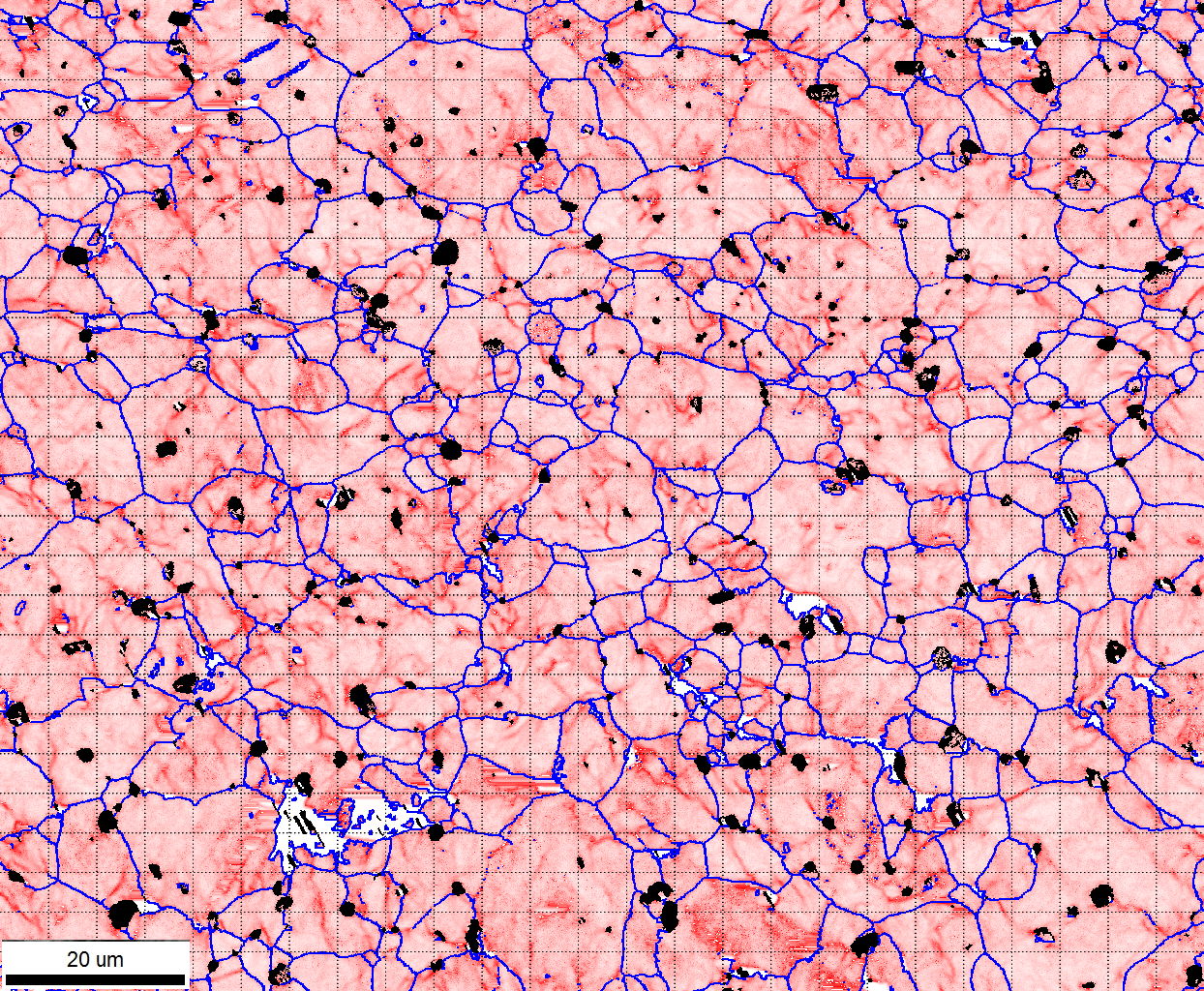}
\caption{Microstructure image showing the KAM at strain level 13.9 \% (overlapped grid of $25\times25$).}
\label{fig:kam25}
\end{figure}

\begin{figure}[!h]
\centering
\includegraphics[scale=0.3]{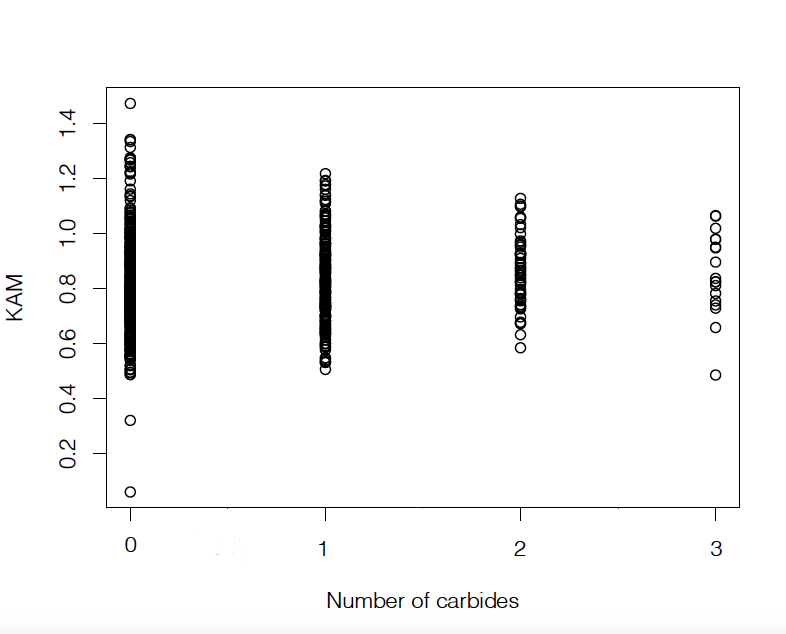}
\caption{Plot of KAM and Number of carbides for the 625 squared areas of Figure \ref{fig:kam25}}
\label{fig:plotdata}
\end{figure}
The results of the estimation in the three different scenarios faced in Sections \ref{subsec:isocase1}-\ref{subsec:isocase3} are summarized in Table \ref{tab:iso15carb}.
\begin{table}[ht]
\centering
\caption{Values of estimated means and variances of the KAM conditioned on the number of carbides visible in a square of a grid $25\times25$ according to different order restrictions assumptions (13.9\% Strain)}
\begin{tabular}{rrrrr}
  \hline
 &0& 1 & 2 & 3  \\ 
  \hline
$\bar y$ & 0.815 & 0.833 & 0.870 & 0.854 \\ 
 $\bar\sigma^2$ & 0.035 & 0.024 & 0.017 & 0.022 \\ 
 $s^2$ & 0.035 & 0.024 & 0.017 & 0.023 \\ 
\hline
  $\hat\mu_{(1)}^I$ & 0.815 & 0.833 & 0.867 & 0.867  \\ 
  \hline
 $\hat\mu_{(2)}^I$ & 0.815 & 0.833 & 0.867 & 0.867  \\ 
 $\hat\sigma_{(2)}$& 0.035 & 0.024& 0.017 & 0.022 \\ 
  \hline
 $\hat\mu_{3)}^I$ & 0.815 & 0.833 & 0.866 & 0.866  \\ 
 $\hat\sigma_{(3)}^I$& 0.035 & 0.024 & 0.018 & 0.018  \\ 
 \hline
  $n$ & 340 & 211 & 54 & 18\\ 
   \hline
\end{tabular}
\label{tab:iso15carb}
\end{table}
\begin{table}[ht]
\centering
\caption{Estimated values for the four different likelihood ratio test with the corresponding parametric and non-parametric p-values}
\begin{tabular}{rrrrr}
  \hline
   &$\hat \mu_{H_0}$&Test Statistic&  p-value (parametric) & p-value (non-parametric) \\ 
$\bar \chi^{2(0)}$&0.827 & 5.760&0.0323&0.0310\\ 
$\bar E^{2(0)}$ & 0.831& 0.0121 & 0.0112 & 0.0085 \\ 
$-2\log\tilde \Lambda$&  0.831 & $7.330$ & 0.0178 & 0.0222 \\ 
$\Lambda^{I(0)}$&  0.831 & $7.105$ & 0.0212& 0.0251  \\ 
\hline
 &0& 1 & 2 & 3  \\ 
 \hline
$\hat\sigma^2_{H_0}$&0.035&0.024&0.018&0.022\\
$\hat\sigma^{2I}_{H_0}$&0.035&0.024&0.019&0.019\\
\hline
\end{tabular}
\label{tab:test15carb}
\end{table}
Modeling the relationship between KAM and carbides and more generally understanding its inhomogeneous distribution over the microstructure is now the main aim and it can be considered a starting point for finding a stochastic model for predicting mechanical properties from 2D microstructure images.
We apply estimation procedures and perform tests under order restrictions, three different univariate isotonic regressions according to the assumption on the variances.\\
The first step for obtaining the data in the most suitable form for the analysis is `overlaying' a grid over the image.
In Figure \ref{fig:kam25} a $25\times25$ grid is added to the image.
With $y_{ij}$, we denote the mean KAM value of the $j$th square of the grid of the image taken in which $i$ carbides are observed.
The explanatory variable $\mathrm{X}$ in all three isotonic regressions is the number of carbides observed in the grid squares. A plot of the data is shown in Figure \ref{fig:plotdata}.\\
We wish to test the null hypothesis that the expected KAM is the same in all the squares of the gird, regardless the numbers of carbides observed in the grid. The alternative hypothesis 
\begin{equation}
\mu_0\le\mu_1\le\mu_2\le\mu_3.
\end{equation}
represents the idea that KAM tends to be higher in areas where more carbides are observed.
Moreover, in the ordered variances case, we assume that 
\begin{equation}
\sigma^2_0\ge\sigma^2_1\ge\sigma^2_2\ge\sigma^2_3>0.
\end{equation}
This is in accordance with what we see in Figure \ref{fig:plotdata}. In fact, the idea behind this assumption is that in areas in which less carbides are observed GNDs have more freedom to move, resulting in increments in dispersion.
In Table \ref{tab:test15carb}, the results of testing the null hypothesis are shown.
For computing $\chi^{2(0)}$ and $\bar E^{2(0)}$, the variance ratio is supposed to be known. In the specific case, we assume that the KAM total variance for the whole image is the real known variance and that $c_i=\frac{\bar\sigma^2_i}{\sigma^2}$.
For computing both the parametric and non-parametric p-values, the two different bootstrap approaches described in Section \ref{sec:bootstrap} have been used and $M$, the number of replications, is taken equal to $20000$.
Independent of the knowledge or assumptions on the variances, the conclusion is the same and leads to the rejection of the null hypothesis.

\section{Conclusions}
\label{sec:conclusions}
This paper presents three different models involving order restrictions and within these models the ML estimators and Likelihood Ratio tests for the homogeneity of the means against monotonicity are introduced and studied.
Prior knowledge given by physical relations or intuition is not often exploited in statistical studies about materials and this can lead to less efficient methods that produce less accurate results. 
After having described the estimation procedures and highlighted how prior knowledge of the variances influence these, we propose the likelihood ratio test as test statistic for testing homogeneity of means.
In the unknown variances case and the ordered variances case, heteroskedasticity plays a crucial role also under the null hypothesis,
leading to different estimates of the common mean under $H_0$.
Results on existence and uniqueness of the maximum likelihood estimates in these last two cases are derived.
Furthermore, two different bootstrap approaches are proposed for approximating the null distribution of the test statistic under the different assumptions on the variances.
The proposed tests are applied to a real data example from Materials Science, showing evidence that the so-called KAM tends to be higher in regions of the microstructure where more carbides are observed.
In fact, incorporating reasonable intuition about the order of means and the variances order in this context helps understanding the evolution of complicated structure of dislocations in metals and its effect on the hardening behaviour of the material during deformation.

\section*{Acknowledgements}
This research was carried out under project number S41.5.14547b in the framework of the Partnership Program of the Materials innovation institute M2i (www.m2i.nl) and the Technology Foundation TTW (www.stw.nl), which is part of the Netherlands Organization for Scientific Research (www.nwo.nl).

\appendix
\section{}
\label{sec:appendix}

\subsection{ALGORITHM 2.1}
\label{subsec:ALGORITHM 2.1}
\begin{description}
\item[0. INITIALIZATION] Let $\bm{\mu}^{(0)}=\bm{\bar y}=(\bar y_1,\dots\bar y_k)'$ and $\bm{w}=(w_1,\dots w_k)'$, $w_i=\frac{n_i}{c_i}$ 
\item[1. QUESTION] Is $\bar y_1\le \bar y_2\le\dots\le \bar y_k$?
\begin{description}
\item[1.1 YES] $\bm{\mu}^*=\bm{\bar y}$ is the solution
\item[1.2 NO] $\bar y_i>\bar y_{i+1}$ \\
Replace $\bar y_i$ and $\bar y_{i+1}$ by
\begin{equation}
m_{i,i+1}=\frac{w_i\bar y_i+w_{i+1}\bar y_{i+1}}{w_i+w_{i+1}}
\end{equation}
Repeat until QUESTION 1 is satisfied.

\end{description}
\end{description}

\subsection{ALGORITHM 2.2}
\label{subsec:ALGORITHM 2.2}
\paragraph*{Two steps Iterative procedure}
\begin{description}
\item[0. INITIALIZATION] Let $\bm{\mu^{(0)}}=\bm{\bar y}=(\bar y_1,\dots\bar y_k)'$, $\bm{\sigma^{2(0)}}=\bm{\bar\sigma^2}=(\bar\sigma^2_1,\dots \bar\sigma^2_k)'$, $\bar\sigma^2_i=\frac{\sum_j(y_{ij}-\bar y_i)^2}{n_i}$ and $\bm{w^{(0)}}=(w^{(0)}_1,\dots w^{(0)}_k)'$, $w_i^{(0)}=\frac{n_i}{\sigma^{2(0)}}$
\item[1. QUESTION] Is $\bar y_1\le \bar y_2\le\dots\le \bar y_k$?
\begin{description}
\item[1.1 YES] $\bm{\mu^*}=\bm{\mu^{(0)}}=\bm{\bar y}$ and $\bm{\sigma^{2}}=\bm{\sigma^{2(0)}}=\bm{\bar\sigma^2} $ are the solutions
\item[1.2 NO]  Use Step 1.2 Algorithm 2.1 to compute $\bm{\mu^{(l)}}$ with weights $\bm{w^{(l-1)}}$
\item[1.2.1] Compute $\bm{\sigma^{2(l)}}=\bm{s^{2(l)}}$, $s_i^{2(l)}=\frac{\sum_j(y_{ij}-\mu_i^{(l)})^2}{n_i}$ and $\bm{w^{(l)}}=(w^{(l)}_1,\dots w^{(l)}_k)'$, $w_i^{(l)}=\frac{n_i}{\sigma^{2(l)}}$
\item[1.2.2] Go back to QUESTION 1 using $\bm{w^{(l)}}$.\\
Repeat until \begin{equation*}
\max_{1\le i\le k}|\mu_i^{*(l-1)}-\mu_i^{*(l)}|\le 10^{-m}
\end{equation*}
\end{description}
\end{description}

\paragraph*{Alternating Iterative Method}
\begin{description}
\item[0. INITIALIZATION] Let $\bm{\nu^{(0)}}=(1/\bar\sigma^2_1,\dots 1/\bar\sigma^2_k)'$, $\bar\sigma^2_i=\frac{\sum_j(y_{ij}-\bar y_i)^2}{n_i}$
\item[1. FIND] $\bm{\mu^(l)}$ the isotonic regression on $D_a$ using weights $\bm{w^{(l-1)}}=(w^{(l-1)}_1,\dots w^{(l-1)}_k)'$, $w_i^{(l-1)}=n_i\nu^{(l-1)}$;
\item[2. FIND] $\bm{\nu^(l)}$ maximizing the profile likelihood $L(y;\bm{\mu^{(l)}};\bm{\nu})$ on $V_0$, $V_0=\{\nu\in\mathbb{R}^k:0\le 1/\max_i(\min_{\min(\bm{\bar y})\le\theta\le\max(\bm{\bar y})}s^2_i(\theta))\le \nu_1\le\dots\le\nu_k\le1/ \min_i(\min_{\min(\bm{\bar y})\le\theta\le\max(\bm{\bar y})}s^2_i(\theta))\}$.\\
Repeat (1)-(2) until
 \begin{equation*}
|L(y;\bm{\mu^{(l-1)}},\bm{\nu^{(l-1)}})-L(y;\bm{\mu^{(l)}},\bm{\nu^{(l)}})|\le 10^{-m}
\end{equation*}
\end{description}
\subsection{ALGORITHM 2.3}
\label{subsec:ALGORITHM 2.3}
\paragraph*{Two steps Iterative procedure}
\begin{description}
\item[0. INITIALIZATION] Let $\bm{\mu^{(0)}}=\bm{\bar y}=(\bar y_1,\dots\bar y_k)'$, $\bm{\sigma^{2(0)}}=\bm{\bar\sigma^2}=(\bar\sigma^2_1,\dots \bar\sigma^2_k)'$, $\bar\sigma^2_i=\frac{\sum_j(y_{ij}-\bar y_i)^2}{n_i}$ and  and $\bm{w^{(0)}}=(w^{(0)}_1,\dots w^{(0)}_k)'$, $w_i^{(0)}=\frac{n_i}{\sigma^{2(0)}}$.
\item[1. QUESTION] Is $\bar y_1\le \bar y_2\le\dots\le \bar y_k$?
\begin{description}
\item[1.1 YES] $\bm{\mu^*}=\bm{\mu^{(0)}}=\bm{\bar y}$ go to QUESTION 2.
\item[1.2 NO] Use Step 1.2 Algorithm 2.1 to compute $\bm{\mu^{(l)}}$ with weights $\bm{w^{(l-1)}}$ 
\item[1.2.1] Compute $\bm{\sigma^{2(l)}}=\bm{s^{2(l)}}$, $s_i^{2(l)}=\frac{\sum_j(y_{ij}-\mu_i^{(l)})^2}{n_i}$ and $\bm{w^{(l)}}=(w^{(l)}_1,\dots w^{(l)}_k)'$, $w_i^{(l)}=\frac{n_i}{\sigma^{2(l)}}$
\item[2. QUESTION]  Is $\sigma^{2(l)}_1\ge \sigma^{2(l)}_2\ge\dots\ge\sigma^{2(l)}_k$?
\begin{description}
\item[2.1 YES] $\bm{\mu^*}=\bm{\mu^{(l)}}$ and $\bm{\sigma^{2*}}=\bm{\sigma^{2(l)}}$ are the solutions.

\item[2.2. NO] $\sigma^{2(l)}_i<\sigma^{2(l)}_j$ \\
Replace $\sigma^{2(l)}_i$ and $\sigma^{2(l)}_{i+1}$ by
\begin{equation}
\sigma^{2(l+1)}_{i}=\sigma^{2(l+1)}_{i+1}=\bar s_{i,i+1}=\frac{n_i\sigma^{2(l)}_i+n_{i+1}\sigma^{2(l)}_{i+1}}{n_i+n_{i+1}}
\end{equation}
Repeat until QUESTION 2 is satisfied.
\item[2.2.1] Go back to QUESTION 1 using $\bm{w^{(l)}}$.\\
Repeat until \begin{equation*}
\max_{1\le i\le k}|\mu_i^{*(l-1)}-\mu_i^{*(l)}|\le 10^{-m} \,\,\,\, \mathrm{and} \,\,\,\, \max_{1\le i\le k}|\sigma_i^{2*(l-1)}-\sigma_i^{2*(l)}|\le 10^{-m} 
\end{equation*}
\end{description}
\end{description}
\end{description}
\paragraph*{Alternating Iterative Method}
\begin{description}
\item[0. INITIALIZATION] Let $\bm{\nu^{(0)}}=(1/\bar\sigma^2_1,\dots 1/\bar\sigma^2_k)'$, $\bar\sigma^2_i=\frac{\sum_j(y_{ij}-\bar y_i)^2}{n_i}$
\item[1. FIND] $\bm{\mu^(l)}$ use Step 1 AlM Algorithm 2.2;
\item[2. FIND] $\bm{\nu^(l)}$ the isotonic regression on $V_0$, $V_0=\{\nu\in\mathbb{R}^k:0\le 1/\max_i(\min_{\min(\bm{\bar y})\le\theta\le\max(\bm{\bar y})}s^2_i(\theta))\le \nu_1\le\dots\le\nu_k\le1/ \min_i(\min_{\min(\bm{\bar y})\le\theta\le\max(\bm{\bar y})}s^2_i(\theta))\}$ with weights $\bm{N}=(n_1,\dots,n_k)'$\\
Repeat (1)-(2) until
 \begin{equation*}
|L(y;\bm{\mu^{(l-1)}},\bm{\nu^{(l-1)}})-L(y;\bm{\mu^{(l)}},\bm{\nu^{(l)}})|\le 10^{-m}
\end{equation*}
\end{description}

\end{document}